\date{}
\documentclass[11pt]{article}
\usepackage{amsmath,amsthm,amsfonts,amssymb,amsopn,amscd}

\def\Ad{{\hbox{\rm Ad}}}

\def\a{\alpha}
\def\b{\beta}

\def\e{\varepsilon}

\def\l{\lambda}

\def\th{\theta}

\def\setminus{\smallsetminus}

\def\A{{\cal A}}
\def\B{{\cal B}}

\def\E{{\cal E}}

\def\M{{\cal M}}
\def\N{{\cal N}}

\def\L{{\cal L}}
\def\I{{\cal I}}
\def\H{{\cal H}}
\def\K{{\cal K}}
\def\S{{\cal S}}

\def\f{{\varphi}}

\def\Si{\mathbb S_{\infty}}

\def\S2{S^{1(2)}}
\newtheorem{theorem}{Theorem}[section]
\newtheorem{lemma}[theorem]{Lemma}

\newtheorem{corollary}[theorem]{Corollary}

\newtheorem{proposition}[theorem]{Proposition}

\theoremstyle{definition} 

\theoremstyle{remark}

\def\setminus{\smallsetminus}

\def\Aa{\A^{(0)}}

\def\Hp{\mathbb H^{\infty}(\mathbb S_{\pi})}
\def\Ha{\mathbb H^{\infty}(\mathbb S_{a})}
\def\RR{{\mathbb R}}

\def\sl2{{{\rm SL}(2,\RR)}}
\def\psl2{{{\rm PSL}(2,\RR)}}

\def\u1{{{\rm V}(1)}}
\def\su2{{{\rm SV}(2)}}

\def\so3{{{\rm SO}(3)}}

\def\A{{\mathcal A}}
\def\B{{\mathcal B}}

\def\F{{\mathcal F}}
\def\H{{\mathcal H}}
\def\I{{\mathcal I}}
\def\K{{\mathcal K}}

\def\M{{\mathcal M}}
\def\N{{\mathcal N}}
\def\O{{\mathcal O}}

\title{ An Algebraic Construction\\ of Boundary Quantum Field Theory}

\author{{\sc Roberto Longo\footnote{Supported in part by the ERC Advanced Grant 227458  
OACFT ``Operator Algebras and Conformal Field Theory", PRIN-MIUR, GNAMPA-INDAM and EU network ``Noncommutative Geometry" MRTN-CT-2006-0031962.}
}
\\
Dipartimento di Matematica,
Universit\`a di Roma ``Tor Vergata'',\\
Via della Ricerca Scientifica, 1, I-00133 Roma, Italy\\
E-mail: {\tt longo@mat.uniroma2.it}
\\
\phantom{X}\\
{\sc Edward Witten}
\\
Institute for Advanced Study,
School of Natural Sciences,\\
Einstein Drive, Princeton, NJ 08540\\
E-mail: {\tt witten@ias.edu}}

\begin{document} 

\maketitle

\begin{abstract}
We build up local, time translation covariant Boundary Quantum Field Theory nets of von Neumann algebras $\A_V$ on the Minkowski half-plane $M_+$ starting with a local conformal net $\A$ of von Neumann algebras on $\mathbb R$ and an element $V$ of a unitary semigroup $\E(\A)$ associated with $\A$. The case $V=1$ reduces to the net $\A_+$ considered by Rehren and one of the authors; if the vacuum character of $\A$ is summable $\A_V$ is locally isomorphic to $\A_+$. We discuss the structure of the semigroup $\E(\A)$. By using a one-particle version of Borchers theorem and standard subspace analysis, we provide an abstract analog of the Beurling-Lax theorem that allows us to describe, in particular, all unitaries on the one-particle Hilbert space whose second quantization promotion belongs to $\E(\A^{(0)})$ with $\A^{(0)}$ the $U(1)$-current net. Each such unitary is attached to a scattering function or, more generally, to a symmetric inner function. We then obtain families of models via any Buchholz-Mach-Todorov extension of $\A^{(0)}$. A further family of models comes from the Ising model. 
\end{abstract}

\newpage

\section{Introduction} 
As is known Conformal Quantum Field Theory is playing a crucial role in several research areas, both in Physics and in Mathematics. Boundary Quantum Field Theory, related to Conformal Field Theory, is also receiving increasing attention.

In recent years, the Operator Algebraic approach to Conformal Field Theory has provided a simple, model independent description of Boundary Conformal Field Theory on the Minkowski half-plane $M_+ =\{\langle t, x\rangle : x>0\}$ \cite{LR1,LR2}.

The purpose of this paper is to give a general Operator Algebraic method to build up new Boundary Quantum Field Theory models on $M_+$. We shall obtain local, Boundary QFT nets of von Neumann algebras on $M_+$ that are not conformally covariant but only time translation covariant.

Motivation for such a construction comes from the papers \cite{W92,W93, LW} where one needs a description of the space of all possible Boundary QFT's in two dimensions compatible with a given theory in bulk. There however a general framework for Boundary QFT was missing and sample computations were given for certain second quantization unitaries $V$ as in the following.

Let us then explain our basic set up. Let $\A$ be a local M\"obius covariant net of von Neumann algebras on the real line; so we have have a von Neumann algebra $\A(I)$ on a fixed Hilbert space $\H$ associated with every interval $I$ of $\mathbb R$ satisfying natural properties: isotony, locality, M\"obius covariance with positive energy and vacuum vector (see Appendix \ref{CN}). We identify the real line with the time-axis of the 2-dimensional Minkowski spacetime $M$. Suppose that $V$ is a unitary on $\H$ commuting with the time translation unitary group, such that $V\A(I_+)V^*$ commutes with $\A(I_-)$ whenever $I_- , I_+$ are intervals of $\RR$ and $I_+$ is contained in the future of $I_-$. Then we can define a local, time translation covariant net $\A_V$ of von Neumann algebras on the half-plane $M_+$ by setting\footnote{If $\M_1$, $\M_2$ are von Neumannn algebras on the same Hilbert space, $\M_1\vee\M_2$ denotes the von Neumann algebra generated by $\M_1$ and $\M_2$. 
}
\[
\A_V(\O)\equiv \A(I_-)\vee V\A(I_+)V^* \ .
\]
Here $\O = I_-\times I_+$ is the double cone (rectangle) of $M_+$ given by $\O\equiv \{\langle t,x\rangle : x\pm t \in I_\pm\}$. The unitaries $V$ as above (that we renormalize for $V$ to be vacuum preserving) form a semigroup that we denote by $\E(\A)$. The case $V=1$ has been studied in \cite{LR1} and gives a M\"obius covariant net $\A_+$ on $M_+$.

So a local, M\"obius covariant net $\A$ and an element $V$ of the semigroup $\E(\A)$ give rise to a Boundary QFT net $\A_V$ on the half-plane. Furthermore, if the split property holds for the local M\"obius covariant net $\A$ (in particular if the vacuum character is summable) the net $\A_V$ is locally isomorphic to the net $\A_+$ on $M_+$. 

Our first problem in this paper is to analyze the structure of $\E(\A)$. We begin by considering the case $\A$ is the net $\A^{(0)}$ generated by the $U(1)$-current and second quantization unitaries, i.e. the unitary on the Fock space one obtains by promoting unitaries on the one-particle Hilbert space. It turns out we are to consider the semigroup $\E(H, T)$ of unitaries $V$ on the one-particle Hilbert space, commuting with the translation one-parameter unitary group $T$, such that $VH\subset H$ where $H$ is the standard real Hilbert subspace associated with the positive half-line (see \cite{LN}).  By using a one-particle version of Borchers theorem and the standard subspace analysis in \cite{LN} we obtain a complete characterization of these unitaries: $V\in\E(H,T)$ if and only if $V=\f(Q)$ with $\f$ the boundary value of a symmetric inner function on the strip $\mathbb S_\pi\equiv \{z: 0 < \Im z < \pi\}$ and $Q$ is the logarithm of the one-particle energy operator $P$, the generator of $T$.

For instance, $\bar T(t) \equiv e^{-it(1/P)}$ gives a one-parameter unitary semigroup in $\E(H,T)$, the only one with negative generator. 

The inner function structure is well known in Complex Analysis and we collect in Appendix \ref{inner} the basic facts needed in this paper. The above result also characterizes the closed real real subspaces $K\subset H$ such that $T(t)K\subset K$, $t\geq 0$, and so is an abstract analog of (a real version of) the Beurling-Lax theorem \cite{B,L} characterizing the Hilbert subspaces of $H^{\infty}(\mathbb S_{\infty})$ mapped into themselves by positive translations in Fourier transform, with $\mathbb S_{\infty}$ the upper complex plane.

Now symmetric inner functions $S_2$ on the strip $\mathbb S_\pi$ with the further symmetry $S_2(-q) = S_2(q)$ are called \emph{scattering functions} and appear in low dimensional Quantum Field Theory (see \cite{Z}); in particular every scattering function will give here a local Boundary QFT net on $M_+$.
One may wonder whether our construction is related to Lechner  models on the 2-dimensional Minkowski spacetime associated with a scattering function \cite{Lechner}, yet at the moment there is no link between the two constructions. 

Our work continues with the construction of local Boundary QFT models associated with other local conformal nets $\A$ on $\RR$. We consider any Buchholz-Mach-Todorov local extension $\A$ of $\A^{(0)}$ (coset models $SO(4N)_1/SO(2N)_2$) \cite{BMT}: every unitary  $V\in\E(\A^{(0)})$, obtained by second quantization of a unitary $V_0=\f(Q)\in\E(H)$ as above, extends to a unitary $\tilde V\in\E(\A)$, provided $\f$ is non-singular in zero. We so obtain other infinite families of local, translation covariant Boundary QFT nets of von Neumann algebras on $M_+$.

A further family of local, translation covariant Boundary QFT nets of von Neumann algebras comes from the Ising model. Also in this case every such model is associated with a symmetric inner function.

\section{Endomorphisms of standard subspaces}
We first recall some basic properties of standard subspaces, we refer to \cite{LN} for more details.

Let $\H$ be a complex Hilbert space and $H\subset\H$ a real linear subspace. The symplectic complement $H'$ of $H$ is the real Hilbert subspace
$H' \equiv \{\xi\in\H : \Im(\xi,\eta)=0\ \ \forall \eta\in H\}$ so $H''$ is the closure of $H$.

A closed real linear subspace $H$ is called cyclic if
$H + iH$ is dense in $\H$ and separating if $H\cap iH=\{0\}$. $H$ is cyclic if and only if $H'$ is separating.  

A {\em standard subspace} $H$ of $\H$ is a closed, real linear subspace of $\H$ which is both cyclic and separating. Thus a closed linear subspace $H$ is standard iff $H'$ is standard.

Let $H$ be a standard subspace of $\H$. Define the anti-linear operator $S \equiv S_H : D(S)\subset\H\to\H$, where $D(S)\equiv H+iH$,
\[
S: \xi  + i\eta\mapsto \xi - i\eta\ ,\quad \xi,\eta\in H\ .
\]
As $H$ is standard, $S$ is well-defined and densely defined, and clearly $S^2 = 1 |_ {D(S)}$. $S$ is a closed operator and indeed its adjoint is given by $S^*_H = S_{H'}$. 
Let
\[
S = J\Delta^{1/2}
\]
be the polar decomposition of $S$.  Then $J$ is an anti-unitary involution, namely $J$ is anti-linear with $J = J^* = J^{-1}$,
and $\Delta \equiv S^*S$ is a positive, non-singular selfadjoint linear operator with
$J\Delta J = \Delta^{-1}$.

The content of following relations is the real Hilbert subspace (much easier) analog of the fundamental Tomita-Takesaki theorem for von Neumann algebras:
\[
\Delta^{it}H = H\ ,\qquad JH =H'\ ,
\]
for all $t\in\RR$.

With $a\in (0,\infty]$ we  denote by $\mathbb S_{a}$ the strip of the complex plane $\{z\in\mathbb C: 0<\Im z < a \}$ (so $\mathbb S_{\infty}$ is the upper plane).
\begin{lemma}\label{incl}
Let $H$ be a standard 
subspace of the Hilbert space $\H$ and $V\in B(\H)$ a bounded linear operator on $\H$. 
The following are equivalent:
\begin{itemize}
\item 
$VH\subset H$ 
\item
$ JVJ\Delta^{1/2}\subset\Delta^{1/2}V$
\item
The map $s\in\RR \to V(s)\equiv\Delta^{-is}V\Delta^{is}$ extends to a bounded weakly continuous function on the closed strip $\overline{\mathbb S_{1/2}}$, analytic in $\mathbb S_{1/2}$, such that
$V(i/2) = JVJ$.
\end{itemize}
\end{lemma}
\begin{proof} 
See \cite{AZ,LN}.
\end{proof}
Note that, if the equivalent properties of Lemma \ref{incl} hold, then $V(s+i/2) = JV(s)J$. Indeed
\begin{equation}\label{is}
V(s + i/2) = \Delta^{iz}V\Delta^{-iz}|_{z=s+ i/2} = \Delta^{iz}V(s)\Delta^{-iz}|_{z= i/2} = JV(s)J\ ,
\end{equation}
where we have applied Lemma \ref{incl} to the unitary $V(s)$.
\medskip

Let $H$ be a standard subspace of the Hilbert space $\H$ and assume that there exists a one parameter unitary group $T(t) = e^{itP}$ on $\H$ such that
\begin{itemize}
\item $T(t)H \subset H$ for all $t\geq 0$
\item $P>0$
\end{itemize}
We refer to a pair $(H,T)$ with $H$ and $T$ as above as a \emph{standard pair} (of the Hilbert space $\H$).  The following theorem is the one-particle analog of Borchers theorem for von Neumann algebras \cite{Borch}.
\begin{theorem}\label{B0}
Let $(H,T)$ be a standard pair as above. The following commutation relations hold for all $t,s\in\RR$:
\begin{equation}\label{B1}
\Delta^{is}T(t)\Delta^{-is} = T(e^{-2\pi s}t)
\end{equation}
\begin{equation}\label{B2}
JT(t)J = T(-t)
\end{equation}
\end{theorem}
\begin{proof} 
See \cite{LN}.
\end{proof}
Note that \eqref{B1} gives a positive energy unitary representation of the translation-dilation group on $\RR$ that we denote by $\L$ ($\L$ is usually called the ``$ax +b$" group): $T(t)$ is the unitary corresponding to the translation $x\mapsto x+t$ on $\RR$ and $\Delta^{is}$ is the unitary corresponding to the dilation $x \mapsto e^{-2\pi s}x$.

We shall say that the standard pair $(H,T)$ is non-degenerate if the kernel of $P$ is $\{0\}$.

Now there exists only one irreducible unitary representation of the group $\L$ with strictly  positive energy, up to unitary equivalence ($\log P$ and $\log\Delta$ satisfies the canonical commutation relations). Therefore, if the standard pair $(H,T)$ is non-degenerate, the associated representation of $\L$ is a multiple of the unique irreducible one and $(H,T)$ is irreducible iff the associated unitary representation of $\L$ is irreducible.

Let the standard pair $(H,T)$ be non-degenerate and (non-zero) irreducible. We can then identify (up to unitary equivalence) $\H$ with $L^2(\RR, {\rm d}q)$, $Q\equiv \log P$ with the operator of multiplication by $q$ on $L^2(\RR, {\rm d}q)$ and $\Delta^{-is}$ by the translation  by $2\pi s$ on this function space:
\begin{equation}\label{identify}
e^{itQ}: f(q)\mapsto e^{itq}f(q),\qquad \Delta^{-is}: f(q)\mapsto f(q+2\pi s)\ .
\end{equation}
In this representation $J$ can be identified with the complex conjugation $Jf =\bar f$ and $f\in H$ iff $f$ admits an analytic continuation on the strip $\mathbb S_{\pi}$ such that $f(\cdot + a)\in L^2$ for every $a\in(0,\pi)$ with boundary values satisfying $f(q+i\pi)=\bar f(q)$.

We now describe the endomorphisms of the standard pair $(H,T)$, namely the \emph{semi-group $\E(H,T)$} of unitaries $V$ of $\H$ commuting with $T$ such that $VH\subset H$ (sometimes abbreviated $\E(H)$). We denote by $P$ the generator of $T$ and begin with the irreducible case.

With $a>0$, we denote by $\Ha$ the space of bounded analytic functions on the strip $\mathbb S_{a}$. If 
$\psi\in\Ha$ then the limit $\lim_{\e\to 0^+}\psi(q + i\e)$ exists for almost all $q\in\RR$ and defines a function in $L^{\infty}(\mathbb R,{\rm d}q)$ that determines $\psi$ (and similarly on the line $\Im z=ia$ if $a<\infty$).
\begin{theorem}\label{endH}
Assume the standard pair $(H, T)$ of $\H$ to be irreducible and let $V$ be a bounded linear operator on $\H$. The following are equivalent:
\begin{itemize}
\item[$(i)$]
$V$ commutes with $T$ and $VH\subset H$;
\item[$(ii)$]
$V = \psi(Q)$ where  $Q\equiv \log P$ and $\psi\in L^{\infty}(\mathbb R,{\rm d}q)$ is the boundary value of a function in $\Hp$  such that $\psi(q + i\pi) = \bar\psi(q)$, for almost all $q\in\RR$.
\end{itemize}
In this case $V$ is unitary, i.e. $V\in\E(H)$, iff $|\psi(q)| =1$ for almost all $q\in\RR$, namely $\psi$ is an inner function on ${\mathbb S}_\pi$, see Appendix \ref{inner}. \footnote{In the scattering context the variable $q$ is usually denoted by $\theta$, the \emph{rapidity}.}
\end{theorem}
\begin{proof}
$(i)\Rightarrow (ii)$:
With $\Delta$ and $J$ the modular operator and the modular conjugation of $H$ we have the commutation relations (\ref{B1},\ref{B2}).
As the standard pair $(H,V)$ is assumed to be irreducible,
the associated positive energy unitary representation of $\L$ is irreducible.

Therefore the von Neumann algebra generated by $\{T(t): t\in\RR\}$ is maximal abelian in $B(\H)$. As $V$ commutes with $T$, setting $Q\equiv \log P$ we have
\[
V = \psi(Q)
\] 
for some Borel complex function $\psi$ on $\RR$. By (\ref{B1},\ref{B2}) we then have
\begin{equation}\label{C1}
\Delta^{-is}\psi(Q)\Delta^{is} = \psi(Q + 2\pi s)
\end{equation}
\begin{equation}\label{C2}
J\psi(Q)J = \bar \psi (Q)
\end{equation}
As $VH\subset H$, by Lemma \ref{incl} and eq. \eqref{is} the function $V(s)\equiv \Delta^{-is}\psi(Q)\Delta^{is}= \psi(Q +2\pi s)$ extends to a bounded continuous function on the strip $\overline{\mathbb S_{1/2}}$, analytic in  $S_{1/2}$, and 
\[
V(s+i/2) = JV(s) J = J\psi(Q+ 2\pi s)J = \bar\psi(Q + 2\pi s)=V(s)^*\ .
\] 
We now identify $\H$ with $L^2(\RR, {\rm d}q)$ with $Q\equiv \log P$ and $\Delta$ as in \eqref{identify}. Then $V$ is identified with the multiplication operator $M_\psi : f\in L^2(\RR, {\rm d}q)\mapsto \psi f \in L^2(\RR, {\rm d}q)$.

It then follows by Lemma \ref{anal} that $\psi$ is the boundary value of a function $\psi\in\Hp$ and $\psi(q+i\pi) = \bar \psi(q)$ for almost all $q\in\RR$.

$(ii)\Rightarrow (i)$: Conversely, let $V=\psi(Q)$ where $\psi$ is the boundary value of a function in $\Hp$ with $\psi(q+i\pi) = \bar \psi(q)$ for almost all $q\in\RR$. Then clearly $V$ commutes with $T$, equations (\ref{C1},\ref{C2}) hold, the function $V(s)= \Delta^{-is}\psi(Q)\Delta^{is}$ is the boundary value of a bounded continuous function on $\overline{\mathbb S_{1/2}}$, analytic on $\mathbb S_{1/2}$, and $V(i/2) = \psi(Q + i\pi) = \bar\psi(Q) =J\psi(Q)J =JVJ$, so $VH\subset H$ by Lemma \ref{incl}. 

Clearly $V$ is unitary iff $|\psi(q)| =1$ for almost all $q\in\mathbb R$. 
\end{proof}
If $(H,T)$ is reducible (and non-degenerate so $Q=\log P$ is defined) the proof of $(ii)\Rightarrow (i)$ in Th. \ref{endH}  remains valid, so the implication still holds true. 

Note that $\E(\A,T)$ is commutative if $(H,T)$ is irreducible: it isomorphic to the semigroup of inner functions. It will be useful to formulate Th. \ref{endH} in terms of functions of $P$ and we do this in the unitary case.
\begin{corollary}\label{fplane}
Let $(H,T)$ be an irreducible standard pair and $V$ a unitary on $\H$.  Then $V\in\E(H,T)$ iff
$V = \f(P)$ with $\f$ the boundary value of a symmetric inner function on 
$\mathbb S_{\infty}$ such that $\f(- p) = \bar\f(p)$, $p \geq 0$.
\end{corollary}
\begin{proof}
Easy consequence of the conformal identification of $\mathbb S_{\infty}$ and $\mathbb S_{\pi}$ by the logarithm function.
\end{proof}
We now describe the Lie algebra of $\E(H,T)$, i.e. the generators of the one-parmeter semigroups of unitaries in $\E(H,T)$.
\begin{corollary}\label{Lie}
Let $(H,T)$ be a standard pair of the Hilbert space $\H$ and $P$ the generator of $T$. Let $V(s)=e^{isA}$ be a one-parameter unitary group of $\H$. 
Then $V(s)\in\E(H,T)$ for all $s\geq 0$ if $A =f(P)$ where $f:\RR\to \RR$ is an odd function $f(-p) = -f(p)$ that admits an analytic continuation in the upper plane $\mathbb S_\infty$ with $\Im f(z) \geq 0$.

Conversely, if $(H,T)$ is irreducible, every unitary one-parameter group $V(s)$ on $\H$ such that $V(s)\in\E(H,T)$ for all $s\geq 0$ has the form $V(s)= e^{isf(P)}$ with $f$ as above.
\end{corollary}
The proof of Corollary \ref{Lie} follows from the analysis in Section \ref{inner} of the semigroup of inner functions; we shall write the explicit form of  $f$ and of the inner functions $\psi$ in Theorem \ref{endH}.
\medskip

\noindent
{\bf Example.} If $(H,T)$ is a non-degenerate standard pair, the self-adjoint operator $-\frac1P$ belongs to the Lie algebra of $\E(H,T)$, namely $e^{-it(1/P)}H\subset H$, $t\geq 0$, with $P$ the generator of $T$.
\medskip

\noindent
We may now describe the reducible case. Let $(\tilde H, \tilde T)$ be a non-zero, non-degenerate standard pair on the Hilbert space $\tilde\H$. Since, up to unitary equivalence, there exists only one non-zero, non-degenerate irreducible standard pair $(H, T)$, the pair $(\tilde H, \tilde T)$ is the direct sum of copies $(H,T)$. In other words we may write 
$\tilde \H = \bigoplus^n_{k=1} \H_k$, $\tilde H = \bigoplus^n_{k=1} H_k$, $\tilde T = \bigoplus^n_{k=1} T_k$, for some finite or infinite $n$, where every Hilbert space $\H_n$ is identified with the same Hilbert space $\H$ and each pair $(H_k, T_k)$ is identified with $(H,T)$. With this identification we have:
\begin{theorem}\label{char2}
A unitary $\tilde V$ belongs to $\E(\tilde H,\tilde T)$ if and only if $\tilde V$ is a $n\times n$ matrix $(V_{hk})$ with entries in $B(\H)$ such that $V_{hk} = \f_{hk}(P)$. Here
$\f_{hk}:\RR\to \mathbb C$ are a complex Borel functions such that $(\f_{hk}(p))$ is a unitary matrix for almost every $p>0$, each $\f_{hk}$ is the boundary value of a function in $\mathbb H(\Si)$ and is symmetric, i.e. $\bar\f_{hk}(p)=\f_{hk}(-p)$.  
\end{theorem}
\begin{proof}
Assume that $V$ belongs to $\E(\tilde H,\tilde T)$.
We may write $\tilde\H = \H\otimes \ell^2$ and $\tilde T = T\otimes 1$; here $\ell^2$ is the Hilbert space of $n$-uples (finite $n$) or of countable square summable sequences ($n=\infty$). As $\tilde V$ commutes with $\tilde T$, $\tilde V$ belongs to the von Neumann algebra $\{T\}'\otimes B(\ell^2)$ which coincides with $\{T\}''\otimes B(\ell^2)$ because $T$ generates a maximal abelian von Neumann algebra. Therefore $\tilde V= (V_{hk})$ where  $V_{hk}=\f_{hk}(P)$ for some complex functions
$\f_{hk}:\RR\to \mathbb C$ and $(\f_{hk}(p))$ is a unitary matrix for (almost) every $p$ because $\tilde V$ is unitary.

Now $\tilde \Delta = \Delta\otimes 1$ and $\tilde J = J\otimes 1$ are constant diagonal matrices, therefore by equations (\ref{B1},\ref{B2}) we have
\begin{equation}\label{C'1}
\Delta^{is}\f_{hk}(P)\Delta^{-is} = \f_{hk}(e^{-2\pi s}P)
\end{equation}
\begin{equation}\label{C'2}
J\f_{hk}(P)J = \bar \f_{hk}(P)
\end{equation}
By Lemma \ref{incl} we have $\tilde\Delta^{1/2} \tilde V\supset\tilde J  \tilde V\tilde J\tilde \Delta^{1/2}$ so
\[
\big(\f_{hk}(-P)\big)=\big(\Delta^{-is}\f_{hk}(P)\Delta^{is}|_{i= 1/2}\big) =\tilde J \tilde  V\tilde J = \big(J\f_{hk}(P) J\big) =\big(\bar\f_{hk}(P)\big)
\]
Therefore 
\[
\bar\f_{hk}(p) = \f_{hk}(-p)\ .
\]
Clearly the matrix operator norm $||\f(z)||$ is bounded, indeed $||\f(z)||\leq 1$.

We may now reverse the above proof to get the converse statement. We only have to check that if each $\f_{hk}(z)$ is bounded then $||\f(z)||\leq 1$. If $n$ is finite this is true because each $\f_{hk}(z)$ is bounded iff $||\f(z)||$ is bounded and in this case $||\f(z)||\leq 1$ by the maximum modulus principle. If $n=\infty$ we then note that the operator norm of each finite corner matrix must be bounded by 1 so $||\f(z)||\leq 1$ also in this case. 
\end{proof}
We note the following proposition: when combined with Cor. \ref{fplane} or Thm. \ref{char2}, it gives an abstract, (real) analog of the Beurling-Lax theorem  \cite{B,L}, see also \cite{Rudin}.
\begin{proposition}\label{subsub}
Let $(H,T)$ be a non-degenerate standard pair of the Hilbert space $\H$. 
A standard subspace $K\subset H$ satisfies $T(t)K\subset K$ for $t\geq 0$ if and only if $K=VH$ for some $V\in\E(H,T)$. 
In particular, if $(H,T)$ is irreducible, $K =\f(P)H$ with $P$ the generator of $T$ and $\f$ a symmetric inner function on $\mathbb S_\infty$.
\end{proposition}
\begin{proof}
Let $U_H$ and $U_K$ be the representations of $\L$ associated with $(H,T)$ and $(K,T)$. By assumptions $U_H$ and $U_K$ agree on the translation one-parameter group, in particular $(K,T)$ is non-degenerate too. Moreover $U_H$ and $U_K$ have the same multiplicity because this is also the multiplicity of the abelian von Neumann algebra generated by $T$. Therefore $U_H$ and $U_K$ are unitarily equivalent and indeed also the associated anti-unitary representations of the group  generated by $\L$ and the reflection $x\mapsto -x$ are unitarily equivalent, namely there exists a unitary $V\in B(\H)$ such that 
\[
U_K(g) = VU_H(g)V^*\ , \ g\in\L \ ,\qquad VJ_KV^* = J_H\ ,
\]
and in particular $V$ commutes with $T(t)$. Then 
\[
VS_H V^* = VJ_H\Delta_H^{1/2} V^* = J_K\Delta_K^{1/2}=S_K\ ,
\] 
hence $VH=K$ and we conclude that $V\in\E(H,T)$. 

The converse statement that if $V\in\E(H,T)$ then $K\equiv VH$ satisfies $T(t)K\subset K$ for $t\geq 0$ is obvious.

In the irreducible case $K=\f(P)H$ by Cor. \ref{fplane}.
\end{proof}
Note that the unitary $V$ in Prop. \ref{subsub} is not unique (but in the irreducible case $V$ is unique up to a sign). On the other hand the unitary
\[
\Gamma\equiv J_K J_H
\]
is a canonical unitary associated with the inclusion $K\subset H$ and commutes with $T$ because $J_H T(t)J_H = T(-t)$ and  $J_K T(t)J_K = T(-t)$, so $\Gamma\in\E(H,T)$. Clearly
\[
\Gamma = VJ_HV^* J_H\ .
\]
In the irreducible case $V=\f(P)$ for some symmetric inner function $\f$ so $J_H V J_H = V^*$ and $\Gamma = V^2$.
\medskip

We now consider the von Neumann algebraic case.
\begin{corollary}\label{endM}
Let $M$ be a von Neumann algebra on a Hilbert space $\H$ with a cyclic and separating vector $\Omega$ and $T(t) =e^{itP}$ a one-parameter unitary group on $\H$, with positive generator $P$, such that $T(t)MT(-t)\subset M$ for all $t\geq 0$. Suppose that the kernel of $P$ is $\mathbb C \Omega$.

If $V$ is a unitary on $\H$ commuting with $T$ such that $VMV^*\subset M$ then $V|_{\H_0}=(\f_{hk}(P_0))$ where $(\f_{hk}(p))$ is a matrix of functions as in Theorem \ref{char2}. Here $\H_0$ is the orthogonal complement of $\Omega$ in $\H$ and $P_0= P|_{\H_0}$.
\end{corollary}
\begin{proof}
With $M_{\rm sa}$ the self-adjoint part of $M$, the closed linear subspace $H\equiv \overline{M_{\rm sa}\Omega}$ is a standard subspace of $\H$ and $VH\subset H$. Thus $(H_0, T)$ is a non-degenerate standard pair of $\H_0$, where  $H_0 = H\ominus \RR \Omega$ and $T_0(t) \equiv T(t)|_{H_0}$.

By Theorem \ref{char2} we then have $V|_{\H_0} = (\f_{hk}(P_0))$ with $(\f_{hk})$ a matrix of functions as in that Theorem. 
\end{proof}
\section{Constructing Boundary QFT on the half-plane}
In this section we introduce the unitary semigroup $\E(\A)$ associated with a local M\"obius covariant net $\A$. By generalizing the construction in \cite{LR1}, each element in $\E(\A)$ produces a Boundary QFT net of local algebras on the half-plane 
\[
M_+\equiv \{\langle t, x\rangle \in\mathbb R^2: x>0\}\ .
\] 
\subsection{The semigroup $\E(\A)$}
Let $\A$ be a local M\"obius covariant net of von Neumann algebras on $\RR$ (see Appendix \ref{CN}); so we have an isotonous map that associates a von Neumann algebra $\A(I)$ on a fixed Hilbert space $\H$
to every interval or half-line $I$ of $\RR$. $\A$ is local namely $\A(I_1)$ and $\A(I_2)$ commute if $I_1$ and $I_2$ are disjoint intervals. 

Denote by $T$ the one-parameter unitary translation group on $\H$. Then $T(t)\A(I)T(-t) =\A(I+t)$, $T$ has positive generator $P$ and $T(t)\Omega =\Omega$ where $\Omega$ is the vacuum vactor, the unique (up to a phase) $T$-invariant vector. By the Reeh-Schlieder theorem $\Omega$ is cyclic and separating for $\A(I)$ for every fixed interval or half-line $I$.
\begin{lemma}\label{defV}
Let $V$ be a unitary on $\H$ commuting with $T$. The following are equivalent:
\begin{itemize}
\item[$(i)$] $V\A(I_2) V^*$ commutes with $\A(I_1)$ for all intervals $I_1 , I_2$ of $\RR$ such that $I_2 > I_1$ ($I_2$ is contained in the future of $I_1$).
\item[$(ii)$] $V\A(a,\infty)V^*\subset \A(a,\infty)$ for every $a\in\RR$.
\item[$(iii)$] $V \A(0,\infty)V^*\subset  \A(0,\infty)$.
\end{itemize}
\end{lemma}
\begin{proof}
Clearly $(ii)\Leftrightarrow (iii)$ by translation covariance as $V$ commutes with $T$. Moreover $(ii)\Rightarrow (i)$ because  $V\A(I_2) V^*\subset  V\A(\tilde I_2) V^*\subset \A(\tilde I_2)$ where $\tilde I_2$ is the smallest right half-line containing $I_2$. Finally, assuming  $(i)$, by additivity we have that $V \A(0,\infty)V^*$ commutes with $\A(-\infty,0)$, so $(iii)$ follows  by duality:  $V \A(0,\infty)V^*\subset\A(-\infty,0)'=\A(0,\infty)$.
\end{proof}
Note that a unitary $V$ in the above Lemma \ref{defV} fixes $\Omega$ up to a phase as it commutes with $T$. We will assume that indeed $V\Omega=\Omega$.
\medskip

\noindent
Let $\A$ be a local M\"obius covariant net of von Neumann algebras on $\mathbb R$ on the Hilbert space $\H$. The unitaries $V$ on $\H$ satisfying the equivalent conditions in Lemma \ref{defV}, normalized with $V\Omega = \Omega$, form a \emph{semigroup} that we denote by $\E(\A)$ (or $\E(\A, T))$.
\medskip

\noindent
Note that $\E(\A,T)\subset \E(H,T)$ where $H\equiv \overline{\A(0,\infty)_{\rm sa}\Omega}$.
As a consequence of Corollary \ref{endM} every unitary $V$ in $\E(\A)$ must have the form $V|_{\H_0}=(\f_{hk}(P_0))$ there described on the orthogonal complement $\H_0$ of $\Omega$.

Examples of unitaries $V$ in $\E(\A)$ are easily obtained by taking either $V$ to implement an internal symmetry (first kind gauge group element), namely $V\A(I)V^* = \A(I)$ for all intervals $I$, or by taking $V=T(t)$ a translation unitary with $t\geq 0$. We give now an example of $V$ in $\E(\A)$ not of this form.
\medskip
  
\noindent {\bf Example.} (See Sect. 8.2 of \cite{BDL}.)
Let $\O$ be a double cone in the Minkowski spacetime $\mathbb R^{d+1}$ with $d$ odd. We denote here by $\A(\O)$ the local von Neumann algebra associated with $\O$ by the $d+1$-dimensional  scalar, massless, free field.

With $I$ an interval of the time-axis $\{{\bf x}\equiv\langle t,x_1,\dots x_d\rangle: x_1=\cdots =x_d = 0\}$ we set
\[
\A_0(I)\equiv \A(\O_I)
\]
where $\O_I$ is the double cone $I''\subset\mathbb R^{d+1}$, the causal envelope of $I$. Then $\A_0$ is a local translation covariant net on $\mathbb R$. (Indeed $\A_0$ extends to a local M\"obius covariant net on $S^1$.) With $U$ the translation unitary group of $\A$, the translation unitary group of $\A_0$ is $T(t)= U(t,0,\dots,0)$.

Let $V\equiv U({\bf x})$ be the unitary corresponding to a positive time-like or light-like translation vector $ {\bf x} = \langle t,x_1,\dots x_d\rangle$ for $\A$, thus $t^2 \geq \sum_{k=1}^d x^2_k$. Then $V\in \E(\A_0, T)$. Indeed $V\A_0(0,\infty)V^* = V\A(V_+)V^* = \A(V_+  + {\bf x})\subset \A(V_+)= \A_0(0,\infty)$, where $V_+$ denotes the forward light cone.

The net $\A_0$ is described as follows:
\[
\A_0 =\bigotimes_{k=0}^{\infty} N_d(k+1)\A^{(k)}
\]
where $\A^{(k)}$ is the local M\"obius covariant net on $S^1$ associated with the $k^{\rm th}$-derivative of the $U(1)$-current and  the multiplicity factor $N_d\big(k + \frac{d-1}{2}\big)$ is the dimension of the space of harmonic spherical functions of degree $k$ on $\RR^d$.
\bigskip

\noindent Before further considerations we characterize the unitaries in $\E(\A)$ implementing internal symmetries.
\begin{proposition}
Let $\A$ be a local M\"obius covariant net of von Neumann algebras on $\RR$ and $U$ the associated unitary representation of the M\"obius group. Then $V\in\E(\A)$ commutes with $U$ if and only if $V$ implements an internal symmetry of $\A$.
\end{proposition}
\begin{proof}
We know that if $V$ implements an internal symmetry then $V$ commutes with $U$ as a consequence of the Bisognano-Wichmann property, see \cite{LN}. Conversely if $V$ commutes with $U$ then $V\A(I)V^*\subset \A(I)$ for every interval $I$ os $S^1$ because the M\"obius group acts transitively on open intervals of $S^1$; in particular also $V\A(I')V^*\subset \A(I')$, thus $V\A(I)V^*\supset \A(I)$ by Haag duality, namely $V$ implements an internal symmetry.
\end{proof}
\subsection{Translation covariant Boundary QFT}
\label{Boundary QFT}
Consider now the 2-dimensional Minkowski spacetime $M$. Let $I_1 , I_2$ be intervals of time-axis such that $I_2 > I_1$ and let $\O= I_1 \times I_2$ be the double cone (rectangle) of $M_+$ associated with $I_1 , I_2$, namely a point $\langle t,x\rangle$ belongs to $\O$ iff $x-t\in I_1$ and $x+t\in I_2$.

We shall say that a double cone $\O= I_1\times I_2$ of $M_+$ is \emph{proper} if it has positive distance from the time axis $x=0$, namely if the closures of $I_1$ and $I_2$ have empty intersection. We shall denote by $\K_+$ the set of proper double cones of $M_+$.
\medskip

\noindent
\emph{A local, (time) translation covariant Boundary QFT net of von Neumann algebras on $M_+$} on a Hilbert space $\H$ is a triple $(\B_+, U, \Omega)$ where
\begin{itemize}
\item $\B_+$ is a isotonous map 
\[
\O\in\K_+\to \B_+(\O)\subset B(\H)
\] 
where $\B_+(\O)$ is a von Neumann algebra on $\H$;
\item
$U$ is a one-parameter group on $\H$ with positive generator $P$ such that
\[
U(t)\B_+(\O)U(-t) = \B_+(\O+\langle t,0\rangle),\quad \O\in\K_+ ,\ t\in\RR;
\]
\item
$\Omega\in\H$ is a unit vector such that $\mathbb C\Omega$ are the $U$-invariant vectors and $\Omega$ is cyclic and separating for $\B_+(\O)$ for each fixed $\O\in\K_+$.
\item $\B_+(\O_1)$ and $\B_+(\O_2)$ commute if $\O_1,\O_2\in\K_+$ are spacelike separated.
\end{itemize}
\subsection{A construction by an element of the semigroup $\E(\A)$}
\label{Vconstruction}
Let now $\A$ be a local, M\"obius covariant net of Neumann algebras on the time-axis $\mathbb R$ of $M_+$. With $V$ a unitary in $\E(\A)$ we set
\[
\A_V(\O) \equiv \A(I_1)\vee V\A(I_2)V^*
\]
where $I_1 , I_2$ are intervals of time-axis such that $I_2 > I_1$ and $\O= I_1 \times I_2$. 
\begin{proposition}\label{AV}
$\A_V$ is a local, translation covariant Boundary QFT net of von Neumann algebras on $M_+$.
\end{proposition}
\begin{proof}
Isotony of $\A_V$ is obvious. Locality means that $\A_V(\O_1)$ commutes elementwise with $\A_V(\O_2)$ if the double cone $\O_2 = I_3 \times I_4$ is contained in the spacelike complement of the double cone $\O_1 = I_1 \times I_2$. Say $\O_2$ is contained in the right spacelike complement of $\O_1$. Then $I_4 > I_2 > I_1 >  I_3$. Now $V\A(I_4)V^*$ commutes with $V\A(I_2)V^*$ by the locality of $\A$ and with $\A(I_1)$ because $V\in\E(\A)$; analogously $\A(I_3)$ commutes with $\A(I_1)$ by locality and with $V\A(I_2)V$ because $V\in\E(\A)$. Therefore $\A(I_3)\vee V\A(I_4)V^*$ and $\A(I_1)\vee V\A(I_2)V^*$ commute. Finally translation covariance with respect to $T$ follows at once because $V$ commutes with $T$ by assumptions.
\end{proof}
If $V=1$ the net $\A_V$ is the net $\A_+$ in \cite{LR1}.
\begin{corollary}\label{uniq}
Let $V_1 , V_2\in\E(\A)$. The following are equivalent:
\begin{itemize}
\item[$(i)$] $\A_{V_1} = \A_{V_2}$;
\item[$(ii)$]  $V_2 = V_1 V$ with $V$ implementing an internal symmetry of $\A$;
\item[$(iii)$]  $V_1\A(0,\infty)V_1^* = V_2\A(0,\infty)V_2^*$,
\end{itemize}
\end{corollary}
\begin{proof}
$(iii) \Leftrightarrow (ii)$ follows by Lemma \ref{defV} and $(ii) \Rightarrow (i)$ is immediate.  
$(i) \Rightarrow (iii)$: note that the von Neumann algebra $V_i\A(-\infty,0)V_i^*$ is generate by the von Neumann algebras $\A_{V_i}(\O)$ as $\O=I_1\times I_2\in\K_+$ varies with $I_1,I_2\subset (-\infty,0)$; therefore $\A_{V_1} = \A_{V_2} \implies 
V_1\A(-\infty,0)V_1^*=V_2\A(-\infty,0)V_2^* \implies V_1\A(0,\infty)V_1^* = V_2\A(0,\infty)V_2^*$ by duality.
\end{proof}
So we have constructed a map:
\[
\text{local M\"ob-covariant net $\A$ on $\mathbb R$}\quad \& \quad  V\in\E(\A)\quad \mapsto\quad \text{BQFT net $\A_V$ on $M_+$} 
\]
and, given $\A$, the map $V\in\E(\A)\mapsto \A_V$ is one-to-one modulo internal symmetries. 
\bigskip

\noindent
We shall say that two nets $\B_1$, $\B_2$ on $M_+$, acting on the Hilbert spaces $\H_1$ and $\H_2$, are \emph{locally isomorphic} if for every proper double cone $\O\in\K_+$  there is an isomorphism $\Phi_\O:\B_1(\O)\to\B_2(\O)$ such that
\[
\Phi_{\tilde\O} |_{\B_1(\O)} = \Phi_\O 
\]
if $\O,\tilde\O\in\K_+$, $\O\subset \tilde\O$ and 
\[
 U_2(t)\Phi_\O(X) U_2(-t) =\Phi_{\O+t}(U_1(t)XU_1(-t)),\quad X\in \B_1(\O)\ ,
\]
with $U_1$ and $U_2$ the corresponding time translation unitary groups on $\H_1$ and $\H_2$.
\begin{proposition}
\label{lociso}
Let $\A$ be a local M\"obius covariant net of Neumann algebras on $\mathbb R$ with the split property. If $V$ is a unitary in $\E(\A)$ the net $\A_{V}$ is locally isomorphic to $\A_+$.
\end{proposition}
\begin{proof} 
Let $I_2 > I_1$ be intervals with disjoint closures and $\O= I_1\times I_2$. Let $\tilde I_2$ be the smallest right half-line containing $I_2$. By the split property there is a natural isomorphism 
\[
\Psi : \A(I_1)\vee \A(\tilde I_2)\to \A(I_1)\otimes \A(\tilde I_2)
\]
with $\Psi(ab) = a\otimes b$ for $a\in\A(I_1)$, $b\in\A(\tilde I_2)$. 

Then the commutative diagram 
\[
\CD
\A_+(\O)\subset   \A(I_1)\vee\A(\tilde I_2)        @>\Psi >>        \A(I_1)\otimes\A(\tilde I_2)    \\
@V{\Phi_\O}VV				             @VV {\rm id}\otimes{{\rm Ad} V} V  \\
\A_V(\O)\subset   \A(I_1)\vee V\A(\tilde I_2)V^*         @>>\Psi >      \A(I_1)\otimes V\A(\tilde I_2)V^* 
\endCD 
\]
defines a natural isomorphism $\Phi_\O:\A_+(\O)\to \A_V(\O)$ and the family $\{\Phi_\O: \O\in\K_+\}$ has the required consistency properties.
\end{proof}
As an immediate consequence, if $V_t$ is a one-parameter semigroup of unitaries in $\E(\A)$, the family $\A_{V_t}$ gives a \emph{deformation} of the conformal net $\A_+$ on $M_+$ with translation covariant nets on $M_+$ that are locally isomorphic to $\A_+$. \medskip

Let again $\A^{(0)}$ be the M\"obius covariant net on $\RR$ associated with by the $U(1)$-current. In other words $\A^{(0)}$ is generated by the $U(1)$-current $j$ 
\[
\A^{(0)}(I) = \left\{W(f)\equiv\exp\left (-i\int j(x)f(x){\rm d}x\right ): {\rm supp}f\subset I\right\}''\ ,
\]
and similarly $\A^{(k)}$ by the net generated by the $k$-derivative of $j$.

Then $\A^{(k)}$ is the net obtained by second quantization of the irreducible, positive energy representation $U^{(k+1)}$ of M\"obius group with lowest weight $k+1$ or, equivalently, $\A^{(k)}$ is the net associated with the irreducible M\"obius covariant net of standard subspaces of the one-particle Hilbert space associated with $U^{(k+1)}$, see \cite{GLW}. 

With $V_0$ a unitary on the one-particle Hilbert space $\H_0$ we denote by $\Gamma(V_0)$ it second quantization promotion to the Bosonic Fock space over $\H_0$.
We shall refer to a unitary of the form $\Gamma(V_0)$ as a \emph{second quantization unitary}.
\begin{theorem} A second quantization unitary $\Gamma(V_0)$ belongs to $\E(\A^{(k)})$ if and only if  $V_0=\f(P^{(k)})$. Here $P^{(k)}$ is the generator of the translation unitary group on the one-particle Hilbert space $\H_0$ and $\f:[0,\infty)\to\mathbb C$ is the  boundary value of a symmetric inner function on $\Si$ as in Corollary \ref{fplane}.
\end{theorem}
\begin{proof}
With $H^{(k)}(0,\infty)$ the standard subspace of $\H_0$ associated with $(0,\infty)$, the von Neumann algebra $\A(0,\infty)$ is generated by the Weyl unitaries $W(h)$ as $h$ varies in $H^{(k)}(0,\infty)$ (see \cite{LN}). As $\Gamma(V_0) W(h) \Gamma(V_0)^* = W(V_0 h)$, we immediately see that $V_0\in\E(H^{(k)}(0,\infty))\Rightarrow \Gamma(V_0) \in\E(\A^{(k)})$.

The converse implication follows because $W(h)\in\A(0,\infty)$ if and only if $h\in H^{(k)}(0,\infty)$ (e.g. by Haag duality). 
\end{proof}
Note that $\Gamma(V_0)$ belongs to a one-parameter semigroup of $\E(\A^{(k)})$ if $\f$ is a singular inner function (see Cor. \ref{siminn}) so to a deformation of the net  $\A^{(k)}_+$; this is not the case if $\f$ is a Blaschke product.  

\subsection{Families of models}

We now construct elements of $\E(\A)$ with $\A$ a local extension of the $U(1)$-current net $\A^{(0)}$; so we get further families of local, translation covariant Boundary QFT nets of von Neumann algebras on $M_+$. For convenience we regard $\A^{(0)}$ as a net on $\RR$.

The local extensions of $\A^{(0)}$ are classified in \cite{BMT}. Such an extension $\A$ is the crossed product of $\A^{(0)}$ by a localized automorphism $\b$. Recall that $\b$ acts on Weyl unitaries by
\[
\b\big(W(h)\big) = e^{-i\int \ell(x)h(x){\rm d}x}W(h)
\]
for every localized element $h$ of the one-particle space, say $h\in\cal S(\mathbb R)$ and $h$ has zero integral, where $S(\mathbb R)$ denotes the Schwartz real function space, see \cite{BMT,GLW}. In other words $\b$ is associated with the action on the $U(1)$-current
\[
j(x)\to j(x) + \ell(x) 
\]
and $\A$ is generated by $\A^{(0)}$ and a unitary $U$ implementing $\b$, see below.

Here $\ell\in \cal S(\mathbb R)$ and the sector class of $\b$ (i.e. the class of $\b$ modulo inner automorphisms) is determined by the charge $g\equiv\frac{1}{2\pi}\int \ell(x){\rm d}x$. $\beta$ is inner iff the charge of $\ell$ is zero and in this case $\beta =\Ad W(L)$ where $L$ is the primitive of $\ell$, namely $L(x)=\int_{-\infty}^x \ell(s){\rm d}s$.

For the extension $\A$ to be local the spin $N= \frac12 g^2$, given by the Sugawara construction, is to be an integer.
We take $\ell$ with support in $(0,\infty)$ so that $\b$ is localized in $(0,\infty)$ and, in particular, $\b$ gives rise to an automorphism of the von Neumann algebra $\A^{(0)}(0,\infty)$. 

As said, $\A^{(0)}$ acts on the Bose Fock space on the one particle Hilbert space $\H$ and $\H$ carries the irreducible unitary representation $U^{(1)}$ of the M\"obius group with lowest weight 1. Therefore we may identify $\H$ with be the Hilbert space $\K_1= L^2(\RR_+,p{\rm d}p)$ with the known lowest weight 1 unitary representation of the M\"obius group; $\cal S(\RR)$ embeds into $\K_1$ (thus in $\H$) by Fourier transform and the scalar product determined by $(f,g) = \int_0^{\infty} p\hat f(p)\overline{\hat g(p)}{\rm d}p$, $f,g\in{\cal S}(\RR)$, see \cite{GLW}. 

Let $H(0,\infty)$ be the standard real Hilbert subspace of $\H$ associated with $(0,\infty)$. Then, in the configuration space, a function $h$ on $\RR$ belongs to  $H(0,\infty)$ if it is real, supp$h\subset [0,\infty)$ and its  Fourier transform $\hat h$ satisfies $\int_0^\infty |p| | \hat h(p)|^2{\rm d}p <\infty $.

Let $\f$ be a symmetric inner function $\f$ on $\Si$ and set $V_0=\f(P)$ with $P$ the positive generator of the time-translation unitary one-parameter group on the one-particle Hilbert space. By Cor. \ref{fplane} the unitary $V_0$ belongs to $\E(H(0,\infty))$. So $V\equiv \Gamma(V_0)\in\E(\A^{(0)})$ and we denote by $\eta$ the endomorphism of  $\A^{(0)}(0,\infty)$ implemented by $V$.

We shall assume that $\frac{|\f(p)-1|^2}{|p|}$ is locally integrable in zero (H\"older continuity at $0$).
Also, as $\f(0)=\pm1$, replacing $\f$ by $-\f$ if necessary we may and do assume that $\f(0)=1$. 

We formally set $\ell_1 = V_0\ell$; rigorously $\ell_1$ is defined to be the function on $\RR$ whose Fourier transform is $\f(p)\hat \ell(p)$. Clearly $\hat\ell(0)=\hat\ell_1(0)$, namely $\int \ell_1(x){\rm d}x  =\int \ell(x){\rm d}x$; moreover $\ell_1$ is real because $\ell$ is real and $\f$ is symmetric.
Note that the support of $\ell_1$ is contained in $[0,\infty)$ by the Paley-Wiener theorem as $\f\in H(\mathbb S_\infty)$. So $\ell -\ell_1$ has zero charge and belongs to $H(0,\infty)$.

In the following $\b$ is the localized automorphism of $\A^{(0)}$ associated with $\ell$. Denote by $L_1$ the primitive of $\ell_1$. Note that the primitive $L-L_1$ of $\ell - \ell_1$ belongs to $H(0,\infty)$, indeed 
\[
\int_0^\infty |p| | \hat L(p) - \hat L_1(p)|^2{\rm d}p  
= \int_0^\infty  \frac{|1- \f(p)|^2}{|p|}|\hat \ell(p)|^2{\rm d}p<\infty \ .
\]
\begin{lemma}\label{be}
On $\A^{(0)}(0,\infty)$ we have
\[
\eta\cdot\beta =\Ad z \cdot \beta \cdot\eta
\]
where the unitary $z$ belongs to $\A^{(0)}(0,\infty)$, indeed $z = W(L - L_1)$.
\end{lemma}
\begin{proof}
For every $h\in H(0,\infty)$ we have
\[
\eta\cdot\beta(W(h)) = \eta\big(e^{-i\int \ell(x)h(x){\rm d}x}W(h)\big) = 
e^{-i\int \ell(x)h(x){\rm d}x}W(V_0 h) 
\]
and 
\begin{multline*}
\Ad z\cdot\beta\cdot\eta(W(h)) = \Ad z\cdot \beta(W(V_0 h))  = e^{-i\int \ell(x)V_0 h(x){\rm d}x}\Ad z\big(W(V_0h)\big)
\\ = e^{-i\int \ell_1(x)V_0 h(x){\rm d}x}W(V_0h) 
= e^{-i\int V_0\ell(x)V_0 h(x){\rm d}x}W(V_0 h) = e^{-i\int \ell(x)h(x){\rm d}x}W(V_0 h)\ .
\end{multline*}
\end{proof}
With $\A$ a local extension of $\Aa$, the von Neumann algebra $\A(0,\infty)$ is generated by $\Aa(0,\infty)$  and a unitary $U$ implementing $\b$, namely $\b(a) = UaU^*$, $a\in\Aa(0,\infty)$; finite sums $\sum_{k=-n}^n a_k U^k$, $a_k\in\Aa(0,\infty)$, are dense in 
$\A(0,\infty)$.
\begin{proposition} 
$\eta$ extends to a vacuum preserving endomorphism $\tilde\eta$ of $\A(0,\infty)$ determined by
$\tilde\eta(U) = zU$ with $z$ as in Lemma \ref{be}.
\end{proposition}
\begin{proof}
Let $\N$ be the subalgebra of $\A(0,\infty)$ of finite sums $\{\sum_k  a_k U^k\}$ with $a_k\in \Aa(0,\infty)$.
It is immediate to check that the map $\tilde\eta_0$
\[
\tilde\eta_0: \sum_k  a_k U^k \mapsto \sum_k \eta(a_k)(zU)^k \ ,\quad a_k\in \Aa(0,\infty)\ ,
\]
is an endomorphism of $\N$. $\tilde\eta_0$ preserves the vacuum conditional expectation $\sum_k a_k U^k \mapsto a_0$, so the vacuum state. Moreover $\tilde\eta_0(\N)$ is cyclic on the vacuum vector $\Omega$ because $\overline{\tilde\eta_0(\N)\Omega}$ contains the closure of $\eta( \Aa(0,\infty))\Omega = V\Aa(0,\infty)\Omega$, namely the Hilbert space of $\Aa$, and is $U$-invariant because $U\in \tilde\eta_0(\N)$.

Then we have a unitary $\tilde V$ determined by
\begin{equation}\label{tildeV}
\tilde VX\Omega = \tilde\eta_0(X)\Omega\ ,\quad X\in\N\ ,
\end{equation}
that implements $\tilde\eta_0$. Therefore $\tilde\eta = \Ad \tilde V$ is a normal extension of $\tilde\eta_0$ to $\Aa(0,\infty)$.
\end{proof}
\begin{proposition} 
The unitary $\tilde V$ defined by eq. \eqref{tildeV} belongs to $\E(\A)$.
\end{proposition}
\begin{proof}
By construction $V$ implements the endomorphism $\tilde\eta$ of $\Aa(0,\infty)$ and $\tilde V\Omega =\Omega$. We only need to show that $\tilde V$ commutes with the translation unitary group $T$ of $\A$, namely $\tilde\eta\cdot \tau_t= \tau_t\cdot \tilde\eta$, $t\geq 0$, on $\Aa(0,\infty)$ with $\tau_t\equiv \Ad T(t)$. Since $V\in\E(\Aa)$, we have $\eta\cdot \tau_t= \tau_t\cdot \eta$ on $\Aa(0,\infty)$ so it suffices to show that  $\eta \tau_t(U)= \tau_t \eta(U)$. 

We have
\begin{equation}\label{coc}
\tau_t(U) = u_t^* U
\end{equation}
where $u_t$ is a unitary $\tau$-cocycle $\Ad u_t \cdot\tau_t \cdot\b= \b\cdot\tau_t$, actually $u_t = W(L - L_t)$ where ${\ell}_t(x) \equiv \ell(x-t)$ and $L_t$ is the primitive of $\ell_t$. Therefore $\tilde\eta \tau_t(U)= \tau_t \tilde\eta(U)$ means $\eta(u_t^*)z U = \tau_t(z)u_t^* U$ and we need to show that $zu_t = \eta(u_t)\tau_t(z)$. Indeed we have
\begin{multline*}
zu_t = W(L_1 - L)W(L - L_t)   = W(L_1 - {L_1}_t)  W( {L_1}_t - L_t) \\
=   W(V_0(L - L_t) )  W({L_1}_t - L_t) =  \eta(u_t)\tau_t(z) \ ,
\end{multline*}
where ${L_1}_t$ is the primitive of ${\ell_1}_t$, so the proof is complete.
\end{proof}
\begin{corollary}
Let  $\f$ be a symmetric inner function  on $\mathbb S_{\infty}$ which is H\"older continuous at $0$ as above with $\f(0)=0$, and $N\in\mathbb N$ be an integer. There is a local, translation covariant Boundary QFT net of von Neumann algebras on $M_+$ associated with $\f$ and $N$.
\end{corollary}
\begin{proof}
Given $N\in\mathbb N$, the extension $\A_N$ of the $U(1)$-current net with charge $g$ such that $\frac12 g^2 = N$ is local \cite{BMT} and $\f$ determines an element $\tilde V\in\E(\A_N)$ as above. Hence we have a Boundary QFT net by the above construction.
\end{proof}
Recall for example the structure of the net $\A_N$ (cf. \cite{BMT}): $\A_1$ is associated with the level 1 $\widehat{su(2)}$-Kac-Moody algebra with central charge 1, $\A_2$ is the Bose subnet of the free complex Fermi field net, $\A_3$ appears in the $\mathbb Z_4$-parafermion current algebra analyzed by Zamolodchikov and Fateev, and in general $\A_N$ is a coset model $SO(4N)_1/SO(2N)_2$.

\subsubsection{Case of the Ising model}
One further family of local Boundary QFT nets comes by considering the Ising  model conformal net $\A_{\rm Ising}$ on $\RR$, namely the Virasoro net with central charge $c= 1/2$.

$\A_{\rm Ising}$ is the fixed point net of $\F$ under the $\mathbb Z_2$ gauge group action, where $\F$ is the twisted-local net of von Neumann algebras on $\RR$ generated by a real Fermi field.

The one-particle Hilbert space of $\F$ is $\H$ carries the irreducible unitary spin $1/2$ representation of the double cover of the M\"obius group and $\F$ acts on the Fermi Fock space over $\H$. 

With $T$ the translation unitary group on $\H$, the standard subspace $H$ of $\H$ associated with $(0,\infty)$ is the one associated with the unique irreducible, non-zero standard pair $(H,T)$.

With $P$ the generator of $T$, then every symmetric inner function $\f$ on $\Si$ gives a unitary $V_0= \f(P)$ on $\H$ mapping $H$ into itself. 

The Fermi second quantization $V$ of $V_0$ then satisfies $V\F(0,\infty)V^*\subset \F(0,\infty)$ and commutes with translations. Moreover $V$ commutes with the $\mathbb Z_2$ gauge group unitary (the Fermi second quantization of $-1$) so it restricts to a unitary $V_-$ on the $\A_{\rm Ising}$ Hilbert subspace and $V_-\A_{\rm Ising}(0,\infty)V_-^*\subset \A_{\rm Ising}(0,\infty)$ namely $V_-\in\E(\A_{\rm Ising})$. By applying our construction we conclude:
\begin{proposition} 
Given any symmetric inner function $\f$ on $\Si$, there is a local, translation covariant Boundary QFT associated with $\A_{\rm Ising}$ as above.
\end{proposition}
\bigskip

\appendix
\noindent{\LARGE \bf Appendix}
\section{One-parameter semigroups of inner functions}\label{inner}
We recall and comment on basic facts about inner functions, see \cite{Rudin}. 

Consider the disk $\mathbb D \equiv \{z\in\mathbb C : |z|<1\}$ and the Hardy space $\mathbb H^{\infty}(\mathbb D)$ of bounded analytic functions on $\mathbb D$.
Every $\f\in \mathbb H^{\infty}(\mathbb D)$ has  a radial limit $f^*(e^{i\th})\equiv\lim_{r\to 1^-}\f(re^{i\th})$ almost everywhere with respect to the Lebesgue measure of $\partial\mathbb D$ and defines a function  $\f^*\in L^{\infty}(\partial\mathbb D,{\rm d}\th)$, where $\partial\mathbb D$ is the boundary of $\mathbb D$. 
As $||\f^*||_{\infty} =\sup \{\f(z): z\in\mathbb D\}$ by the maximum modulus principle, and in particular $\f^*$ determines $\f$, we may identify $\mathbb H^{\infty}(\mathbb D)$ with a Banach subspace of $L^{\infty}(\partial\mathbb D,{\rm d}\th)$. We shall then denote $\f^*$ by the same symbol $\f$ if no confusion arises.

Given a sequence of elements $a_n\in\mathbb D$ such that $\sum_{n=1}^{\infty} (1 - \lvert a_n \rvert) < \infty$ , the function
\begin{equation*} 
B(z) \equiv \prod_{n=1}^\infty  B_{a_n}(z) 
\end{equation*}
is called the \emph{Blaschke product}. Here $B_a(z)$ is the Blaschke factor $\frac{\lvert a \rvert}{a}\frac{z-a}{1-\bar{a}z}$ if $a\neq 0$ and $B_0(z)\equiv z$.
This product converges uniformly on compact subsets of the $\mathbb D$, and thus $B$ is a holomorphic function on the disk. Moreover $\left\lvert B(z) \right\rvert \leq 1$ for $z\in\mathbb D$.

An \emph{inner function} $\f$ on $\mathbb D$ is a function $\f\in \mathbb H^{\infty}(\mathbb D)$ such that 
$|\f(z)|=1$ for almost all $z\in\partial\mathbb D$.\footnote{Every function in $\mathbb H^{\infty}(\mathbb D)$ factorizes into the product of an inner function and an outer function \cite{Rudin}. We don't need this fact here.}
A Blaschke product is an inner function. Indeed, up to a phase, $B_a(z)$ is the only inner function with a simple zero in $a$ (thus the M\"obius transformation mapping $a$ to $0$) and $B(z)$ the only inner function on $\mathbb D$  that has zeros exactly at $\{ a_n \}$, with multiplicity. 

If an inner function $\f$ has no zeros on $\mathbb D$, then $\f$ is called a \emph{singular} inner function.
  
$\f$ is an inner function if and only if  
\begin{equation}\label{innereq} 
\f(z) = \alpha B(z) \exp\left( - \int_{-\pi}^\pi \frac{e^{i\theta}+z}{e^{i\theta}-z}{\rm d}\mu(e^{i\theta}) \right) , 
\end{equation}
where $\mu$ is a positive, finite, Lebesgue singular measure on $\partial\mathbb D$, $B(z)$ is a Blaschke product and $\alpha$ is a constant with $\lvert \alpha \rvert = 1$. The decomposition is unique.
Note that all the zeros of $\f$ come from the Blaschke product so $\f$ is singular if and only if $B$ is the identity.

Note that the inner functions form a (multiplicative) {\it semigroup} and the singular inner functions form a sub-semigroup. We now consider a one-parameter semigroup $\{\f_t, t\geq 0\}$ of inner functions. Namely $\f_t$ is an inner function for every $t\geq 0$, and $\f_{t+s} =\f_t \f_s$. Clearly $\f_0 = 1$. We require that the map $t\in [0,\infty)\to \f^*_t$ is weak$^*$ continuous in $L^{\infty}(\partial\mathbb D,{\rm d}\th)$. This is equivalent to the weak operator continuity (hence to the strong operator continuity) of the one-parameter unitary group $M_{\f^*_t}$ on $L^{2}(\partial\mathbb D,{\rm d}\th)$, where $M_{\f^*_t}$ is the multiplication operator by $\f^*_t$ on $L^{2}(\partial\mathbb D,{\rm d}\th)$.

\begin{proposition}\label{innchar}
Let $\f_t$ be a one-parameter semigroup of inner functions on $\mathbb D$. Then:

$a)$ $\f_t(z)\to 1$ as $t\to 0$ uniformly on compact subsets of $\mathbb D$,

$b)$ every $\f_t$ is singular,

$c)$ $\f_t(z) = e^{itf(z)}$ where $f$ is a analytic function on $\mathbb D$ with $\Im f(z) \geq 0$ such that the radial limit function of $f$ on $\partial\mathbb D$ exists almost everywhere and is real.
\end{proposition}
\begin{proof}
$a)$: By the weak$^*$ continuity of $\f_t$ we have
\begin{equation}\label{int}
\oint_{\partial\mathbb D}\f_t(z) g(z){\rm d}z \longrightarrow \oint_{\partial\mathbb D}g(z){\rm d}z
\end{equation}
as $t\to 0$, for all $g\in L^1(\partial\mathbb D,{\rm d}\th)$.

Let $z_0\in\mathbb D$. Since $\f_t\in \mathbb H^{\infty}(\mathbb D)$ the value  $\f_t(z_0)$ is given by the Cauchy integral
\[
\f_t(z_0) = \frac{1}{2\pi i}\oint_{\partial\mathbb D}\frac{\f_t(z)}{z-z_0}{\rm d}z
\]
so, choosing $g(z)\equiv \frac{1}{2\pi i}\frac{1}{z-z_0}$ in \eqref{int}, we see that $\f_t(z_0)\to 1$ as $t\to 0$. As the family of analytic functions $\{\f_t : t >0\}$ is bounded, hence normal, the convergence is indeed uniform on compact subsets of $\mathbb D$.

$b)$: Fix $z_0\in\mathbb D$ and suppose $z_0$ is a zero of some $\f_t$. Let $t_0\equiv \inf\{t>0: \f_t(z_0)=0\}$. Since $\f_t(z_0)\to 1$ as $t\to 0$ we have $t_0 >0$. Write now $t= ns$ with $s\in (0,t_0)$ and $n$ an integer. Then 
$\f_t(z_0)=\f_{ns}(z_0)= \f_s(z_0)^n\neq 0$, so we conclude that $\f_t$ never vanishes in $\mathbb D$ for every $t>0$.

$c)$: For a fixed $z\in\mathbb D$, the map $t\mapsto \f_t(z)$ is a one-parameter semigroup of complex numbers with modulus less than one, therefore $\f_t(z) = e^{itf(z)}$ for a complex number $f(z)$ such that $\Im f(z)\geq 0$. 

Now, by point $a)$, on any given compact subset of $\mathbb D$, we have $|\f_t(z) - 1| <1$ for a sufficiently small $t>0$; thus $itf(z) = \log\f_t(z)$ showing that $f$ an analytic function on $\mathbb D$. This also shows that $f(z)$ has a real radial limit to almost all points of $\partial\mathbb D$.
\end{proof}
We shall say that $\f\in \mathbb H^{\infty}(\mathbb D)$ is \emph{symmetric} if $\f(z) = \bar\f(\bar z)$ for all $z\in \mathbb D$, thus iff $\f$ is real on the interval $(-1,1)$ ($\f$ is real analytic). Of course $\f$ is symmetric iff the equality $\f(z) = \bar\f(\bar z)$ holds almost everywhere on the boundary $\partial\mathbb D$. Note that a Blaschke factor $B_a$ is symmetric iff $a$ is real, thus a Blaschke product is symmetric iff the non-real zeros come in pairs, with multiplicity.

We now determine all semigroups of inner functions.
\begin{corollary}\label{siminn}
Every one-parameter semigroup of inner functions $\f_t$ on $\mathbb D$ is given by
\begin{equation}\label{innersemi}
\f_t(z) = e^{it\l}\exp\left( - t\int_{-\pi}^\pi \frac{e^{i\theta}+z}{e^{i\theta}-z}{\rm d}\mu(e^{i\theta}) \right) , 
\end{equation}
where $\mu$ is a positive, finite measure on $\partial\mathbb D$ which is singular with respect to the Lebesgue measure and $\l\in\RR$ is a constant. 

Conversely, given a finite, positive, Lebesgue singular measure $\mu$ and a constant $\l\in\RR$, formula \eqref{innersemi} defines a one-parameter semigroup of inner functions on $\mathbb D$.

All functions $\f_t$ are symmetric if and only if $\l =0$ and $\mu(e^{i\theta}) =\mu(e^{-i\theta})$.  
\end{corollary}
\begin{proof}
Let $\f$ be a semigroup of inner functions $\f_t$ on $\mathbb D$. By point $c)$ in Prop. \ref{innchar} every $\f_t$ is singular for every $t\geq 0$. By formula \eqref{innereq} we then have
\[
\f_t(z) = \a(t)\exp\left( - \int_{-\pi}^\pi \frac{e^{i\theta}+z}{e^{i\theta}-z}{\rm d}\mu_t(e^{i\theta}) \right)\ ,
\]
where $\a(t)$ is complex numbers of modulus one and $\mu_t$ is a Lebesgue singular measure.
Clearly $\a$ is a semigroup, so $\a(t) = e^{it\l}$ for a real constant $\l$. By comparing the above expression with the formula $\f_t = e^{itf}$ given by point $c)$ in Prop. \ref{innchar}, we see that $\frac{\mu_t}{t}$ is a constant, namely $\mu_t = t\mu$ for a Lebesgue singular measure $\mu$ as desired.

The rest is immediate.
\end{proof}
Therefore if $\f$ is a symmetric inner function then:
\[
\text{$\f$ belongs to a one-parameter semigroup of symmetric inner functions} \ \Leftrightarrow \ \text{$\f$ is singular.}
\]
Set $h(z)\equiv i\frac{1+z}{1-z}$. We now use the conformal maps $h$ and $\log$ to identify $\mathbb D$ with $\mathbb S_{\infty}$ and with $\mathbb S_{\pi}$ as follows
\[
\mathbb D \overset{h}{\longrightarrow}\mathbb S_{\infty}\overset{\log}{\longrightarrow}\mathbb S_{\pi}\ .
\]
With this identification we shall carry the above notions to $\mathbb S_{\infty}$ and $\mathbb S_{\pi}$. In particular given a function $\f\in \mathbb H^{\infty}(\mathbb S_{\pi})$ (resp. $\f\in \mathbb H^{\infty}(\mathbb S_{\infty})$) we shall say that $\f$ is \emph{symmetric} iff $\f(q+i\pi) = \bar\f(q)$ (resp. $\f(-q) = \bar\f(q)$) for almost all $q\in\RR$; and $\f$ is \emph{inner} if $|\f(q)| = |\f(q+i\pi)| =1$ for almost all $q\in\RR$ (resp. $|\f(q)| =1$ for almost all $q>0$). 

Note that by eq. \eqref{innereq} every inner function $\f$ on $\Si$ can be uniquely written  as
\begin{equation}\label{siminn2}
\f(p) = B(p) \exp\left( - i\int_{-\infty}^{+\infty} \frac{1+ps}{p - s}{\rm d}m(s) \right) . 
\end{equation}
Here $m$ is a measure on $\RR\cup\{\infty\}$ singular with respect to the Lebesgue measure (the point at infinity can have positive measure). A factor in the Blasckhe product $B$ here have the form $\frac{p-\a}{p+\a}$ with $\Im \a\geq 0$ and is symmetric iff $\Re\a =0$. Clearly $\f$ is symmetric iff the Blasckhe factors corresponding to $\a$ and $-\bar\a$, with $\Re\a \neq 0$, appear in pairs (with the same multiplicity) and $m(s) = m(-s)$.
\begin{corollary}\label{onepar}
Let $\f_t$, $t\geq 0$, be a semigroup of symmetric inner functions on $\mathbb S_{\infty}$. Then 
$\f_t(z) = \exp(itf(z))$ where $f$ is holomorphic on $\mathbb S_{\infty}$ with $\Im f(z)\geq 0$. For almost all $p\in\RR$ the limit $f^*(p) = \lim_{\e\to 0^+}f(p+i\e)$ exists almost everywhere and is real with $f^*(-p) = -f^*(p)$. For $p\geq 0$ we have:
\begin{equation}\label{generator}
f^*(p) = cp + \int_0^{+\infty}\!\frac{p}{\l^2 - p^2}{\rm d}\nu(\l)
\end{equation}
where $c\geq 0$ is a constant and $(1+\l^2)^{-1}\nu(\l)$ is a finite, positive measure on $[0,+\infty)$ which is singular with respect to the Lebesgue measure.

Conversely every function $f^*$ on $[0,\infty)$ given by the right hand side of \eqref{generator} is the boundary value of a function $f$ analytic in $\mathbb S_{\infty}$ and $\f_t \equiv e^{itf}$, $t\geq 0$, is a one-parameter semigroup of symmetric inner functions on $\mathbb S_{\infty}$.
\end{corollary}
\begin{proof} 
This is a consequence of Cor. \ref{siminn} and formula \eqref{siminn2}.
\end{proof}
With $\psi\in L^{\infty}(\RR,{\rm d}q)$, denote by $M_{\psi}$ 
the operator of multiplication by $\psi$ on $L^2(\RR,{\rm d}q)$. Set also $\psi_s(q) = \psi(q+ s)$.
\begin {lemma}\label{anal}
Let $\psi\in L^{\infty}(\RR,{\rm d}q)$. The operator-valued map $s\in\RR\to V(s)\equiv M_{\psi_s}\in B(L^{2}(\RR,{\rm d}q))$ extends to a bounded weakly continuous function on the strip $\overline{\mathbb S_{a}}$, $a>0$, analytic in $\mathbb S_{a}$, such that $V(s+ia) = V(s)^*$ if and only if $\psi$ is the boundary value of a function in $\mathbb H^{\infty}(\mathbb S_{a})$ such that $\psi(s+ia) =\bar\psi(a)$ for almost $s\in\mathbb R$.
\end{lemma}
\begin{proof}
Suppose that  $s\in\RR\to V(s)\in B(L^{2}(\RR,{\rm d}q))$ extends to a bounded weakly continuos function on the strip $\overline{\mathbb S_{a}}$, analytic in $\mathbb S_{a}$, and $V(s+ia) = V(s)^*$. Then for every $g\in L^{1}(\RR,{\rm d}q)$ the map
\[
s\in\RR\to (g_1, V(s) g_2) = \int\psi(q+s)g(q){\rm d}q
\]
is the boundary value of a function $V_g$ in $\mathbb H^{\infty}(\mathbb S_a)$ such that $V_g(s+ia) = V^*_{\bar g}$. Here $g_1 , g_2$ are $L^2$-functions with $g_1\bar g_2 = g$.
For a fixed $u\in (0,a)$ the map $g\to V_g(iu)$ is a linear functional on $L^{\infty}(\RR,{\rm d}q)$ which is weak$^*$ continuos by the maximum modulus theorem. Thus $V_g(iu) = \int \psi_{iu}(q)g(q){\rm d}q$  with $\psi_{iu}$ a $L^{\infty}$-function. Setting $\psi(z)= \psi_{iu}(q)$ with $z= q+iu$ one can then show that $\psi $ is
a function in $\mathbb H^{\infty}(\mathbb S_{a})$ and $\psi(s+ia) =\bar\psi(a)$.

The converse statement is easily verified.
\end{proof}
By a \emph{scattering function} $S_2$ we shall mean a symmetric inner function on $\mathbb S_\pi$ which is continuous on $\overline{\mathbb S_\pi}$ with the additional symmetry $S_2(-\bar z)= S_2(z)$ (cf. \cite{Lechner}).

Let $\f$ be an inner function on $\mathbb S_\pi$ which is continuous on $\overline{\mathbb S_\pi}$. Viewed as a function on $\overline{\mathbb D}$, $\f$ has only two possible singularities at $1$ and $-1$; if it is further singular then by eq. \eqref{innersemi} 
\[
\f(z)=\exp\left(c_1\frac{z+1}{z-1} - c_2\frac{z-1}{1+z} \right)\ . 
\]
for some constants $c_2\geq 0$, $c_2\geq 0$ and $\f$ is a scattering function iff $c_1 = c_2$.
\begin{corollary}
Let $\f_t$ be a one-parameter semigroup of symmetric inner functions on $\mathbb S_\infty$ and let $f$ be its generator, i.e. $\f_t(z) = e^{itf(z)}$. The following are equivalent:
\begin{itemize}
\item $f$ is holomorphic in $\mathbb C$ with at most one singularity in $0$,
\item $f(z) = c_1 z - c_2 \frac1z$
for some constants $c_2\geq 0$, $c_2\geq 0$,
\item viewed as a function on $\mathbb S_\pi$, $\f_t$ is continuous up to the boundary for every $t\geq0$.
\end{itemize}
In particular $\f_t$ is a scattering function for every $t\geq 0$ iff $f(z) = c(z - \frac1z)$ with $c\geq 0$.
\end{corollary}
\begin{proof}
Immediate by the above discussion.
\end{proof}
\section{M\"obius covariant nets of von Neumann algebras}
\label{CN}
The reader may find in \cite{LN} the basic properties of local, M\"obius covariant net of von Neumann algebras on $\mathbb R$. Here we recall the definition.
\smallskip

A \emph{net $\A$ of  von~Neumann algebras on $S^1$} is a map
\[
I\to\A(I)
\]
from $\I$, the set of open, non-empty, non-dense intervals of $S^1$, to the set of von Neumann algebras on a (fixed) Hilbert space $\H$ that verifies the following isotony property:
\begin{description}
\item[$\textnormal{\textsc{1. Isotony}}:$] {\it If $I_1$, $I_2$ are intervals and $I_1\subset I_2$, then
\[
\A(I_1)\subset\A(I_2)\ .
\]}
\end{description}
The net $\A$ is said to be  \emph{M\"obius covariant} if the following properties 2,3 and 4 are satisfied:
\begin{description}
\item[\textnormal{\textsc{2. M\"obius invariance}}:] {\it There is a
strongly continuous unitary representation $U$ of ${\bf G}$ on $\H$ such that
\[
U(g)\A(I)U(g)^*=\A(gI)\ , \quad g\in {\bf G},\ I\in\I.
\]}
\end{description}
Here $\bf G$ denotes the M\"obius group (isomorphic to $PSL(2,\mathbb R)$) that naturally acts on $S^1$.
\begin{description}
\item[$\textnormal{\textsc{3. Positivity of the energy}}:$]
{\it $U$ is a positive energy representation.} 
\end{description}
\begin{description}
\item[\textnormal{\textsc{4. Existence and uniqueness of the vacuum}}:]
{\it There exists a unique (up to a phase) unit $U$-invariant vector $\Omega$ 
(vacuum vector) and $\Omega$ is
cyclic for the von Neumann algebra $\vee_{I\in\I}\A(I)$}
\end{description}
The net $\A$ is said to be  \emph{local} if the following property holds:
\begin{description}
\item[$\textnormal{\textsc{5. Locality}}:$]
{\it If $I_1$ and $I_2$ are disjoint intervals, the von Neumann algebras $\A(I_1)$ and $\A(I_2)$ commute:}
\[
\A(I_1)\subset\A(I_2)'
\]
\end{description}
A \emph{local M\"obius covariant net on $\RR$} is the restriction of a local M\"obius covariant net on $S^1$ to $\RR = S^1\setminus \{-1\}$ (identification by the stereographic map).
\smallskip

We say that the \emph{split} property holds for a local M\"obius covariant net $\A$ on $S^1$ if $\A(I_1)\vee\A(I_2)$ is naturally isomorphic with $\A(I_1)\otimes\A(I_2)$ when $I_1$, $I_2$ are intervals with disjoint closures. (If $\A$ is non-local one requires that the inclusion $\A(I_1)\subset \A(I'_2)$ has an intermediate type I factor.) This very general property holds in particular if Tr$(e^{-\b L_0}) <\infty$ for all $\b>0$, where $L_0$ is the conformal Hamiltonian, see \cite{BDL}.

\bigskip

\noindent
{\bf Acknowledgements.} The first named author is grateful to K.-H.
Rehren for comments.

\end{document}